\documentclass{article}

\usepackage[affil-it]{authblk}
\usepackage[dvipsnames]{xcolor}
\usepackage{amsfonts}
\usepackage{amsmath,amsthm,amssymb,dsfont,centernot}

\usepackage{enumerate}
\usepackage{graphicx}	
\usepackage{subcaption}
\usepackage[margin=3cm]{geometry}
\usepackage{url}
\usepackage{todonotes}
\usepackage{bbm}

\usepackage{tikz}

\usepackage{booktabs}

\usepackage{pifont}
\usepackage{multirow}
\usepackage{makecell}

\usepackage{epsfig}
\usetikzlibrary{shapes.symbols,patterns} 
\usepackage{pgfplots}
\pgfplotsset{compat=1.10}
\usepgfplotslibrary{fillbetween}

\definecolor{linkblue}{HTML}{001487}
\usepackage{hyperref}
\hypersetup{
  colorlinks=true,
  citecolor=linkblue,
  linkcolor=linkblue,
  filecolor=linkblue,
  urlcolor=linkblue,
  breaklinks=true,
  pdftitle={Relative entropy representations via tracial joint spectral measures},
  pdfauthor={Lukas Schmitt}
}

\usepackage{nicefrac}
\usepackage{mathtools}

\usepackage{thmtools}
\hypersetup{hypertexnames=false}

\usepackage{algorithm}
\usepackage{algorithmic}

\usepackage{mdframed}
\usepackage{aligned-overset}
\usepackage{circuitikz}

\theoremstyle{plain}
\newtheorem{theorem}{Theorem}[section]
\newtheorem{lemma}[theorem]{Lemma}

\newtheorem{corollary}[theorem]{Corollary}
\newtheorem{proposition}[theorem]{Proposition}

\theoremstyle{definition}

\newtheorem{remark}[theorem]{Remark}
\newtheorem{example}[theorem]{Example}

\usepackage{tcolorbox}
\tcbuselibrary{breakable}
\tcbset{breakable}
\tcolorboxenvironment{definition}{}
\tcolorboxenvironment{theorem}{colback=pink!25!white,colframe=pink!100!black}
\tcolorboxenvironment{lemma}{colback=yellow!5!white,colframe=yellow!75!black}
\tcolorboxenvironment{corollary}{colback=Dandelion!5!white,colframe=Dandelion!75!black}
\tcolorboxenvironment{proposition}{colback=Emerald!5!white,colframe=Emerald!75!black}
\tcolorboxenvironment{claim}{colback=RoyalBlue!5!white,colframe=RoyalBlue!75!black}
\tcolorboxenvironment{conjecture}{colback=red!5!white,colframe=red!75!black}
\tcolorboxenvironment{example}{colback=white,colframe=lightgray}
\tcolorboxenvironment{remark}{colback=white,colframe=lightgray}
\tcolorboxenvironment{question}{colback=lightgray!5!white,colframe=lightgray!75!black}

\DeclareRobustCommand{\abbrevcrefs}{%
\Crefname{theorem}{Thm.}{Thms.}%
\Crefname{corollary}{Cor.}{Cors.}%
\Crefname{lemma}{Lem.}{Lems.}%
\Crefname{remark}{Rmk.}{Rmks.}%
\Crefname{proposition}{Prop.}{Props.}%
\Crefname{equation}{Eq.}{Eqs.}%
\Crefname{example}{Ex.}{Exs.}%
}

\DeclareRobustCommand{\cshref}[1]{{\abbrevcrefs\cref{#1}}}

\newcommand{\nll}{\centernot{\ll}}

\newcommand*{\ee}{\mathrm{e}}

\newcommand*{\cH}{\mathcal{H}}

\newcommand*{\cT}{\mathcal{T}}

\newcommand*{\R}{\mathbb{R}}

\newcommand*{\id}{\mathds{1}}

\newcommand*{\spec}{\mathrm{spec}}
\newcommand*{\supp }{\mathrm{supp}}
\newcommand*{\tr}{\mathrm{tr}}
\newcommand*{\ket}[1]{| #1 \rangle}
\newcommand*{\bra}[1]{\langle #1 |}

\newcommand{\proj}[1]{|#1\rangle\!\langle #1|}

\newcommand*{\D}{\mathrm{D}}

\newcommand*{\di}{\mathrm{d}} 

\newcommand{\norm}[1]{\left\lVert#1\right\rVert}

\newcommand*{\meas}{\, \di \mu_{A,B}(a,b)}

\usepackage{float}
\usepackage[nameinlink,capitalize,noabbrev]{cleveref}

 \allowdisplaybreaks

\title{Relative entropy representations via tracial joint \\ spectral measures }
\author{Lukas Schmitt$^{1,2}$}
\affil{\small $^{1}$Institute for Theoretical Physics, ETH Zurich\\
  $^{2}$IBM Quantum, IBM Research Europe -- Zurich
 }
 \date{}

\begin{document}

\maketitle
\begin{abstract}
We show that for positive semidefinite $A$ and $B$ the relative entropy can be written as $D(A\|B)=2\int a\log(a/b)\,\di\mu_{A,B}(a,b)$, where $\mu_{A,B}$ denotes Heinävaara's tracial joint spectral measure. 
This allows us to obtain  a unified derivation of several known integral representations of relative entropy from scalar equalities. Scalar inequalities can be used in the same way to derive Pinsker-type bounds. 
We further show that the testing curve $t\mapsto\tr[A-tB]_+$ is piecewise affine exactly when $A$ and $B$ commute, which is equivalent to their testing region being a polygon. This gives a characterization of noncommutativity through the curvature. 
As further applications, we derive inequalities for relative entropy variance and the loss of relative entropy under positive trace-preserving maps.
\end{abstract}

\section{Introduction}
\label{sec:intro}

For positive semidefinite matrices $A,B\in M_n(\mathbb C)$, the quantum relative entropy is given by
\begin{equation}\label{eq:Ddef}
    D(A\|B) \coloneqq \tr[A(\log A - \log B)] \, ,
\end{equation}
with $D(A\|B)=+\infty$ unless $\supp \, A\subseteq\supp \, B$. Recently, many useful integral representations of this quantity were found, expressing it in terms of the expression $A-sB$ \cite{Frenkel_2023,Hirche_2024,liu2025layercakerepresentationsquantum}. 
Although these formulas can vary a lot, we show that they can be related by scalar equalities using Heinävaara's tracial joint measure.
Heinävaara~\cite{Heinavaara2023} assigns to every Hermitian pair $(A,B)$ a positive measure $\mu_{A,B}$ on $\mathbb R^2$ such that
\begin{equation}\label{eq:intro_master}
    \tr\bigl[H(f)(xA+yB)\bigr] = \int_{\mathbb{R}^2} f(ax+by)\, \di\mu_{A,B}(a,b) 
    \quad \textnormal{with} \quad H(f)(z) \coloneqq \int_0^1 \frac{1-t}{t}\, f(zt)\, \di t \,,
\end{equation}
for a large class of functions $f$ and every $x,y\in\mathbb R$ (see \cref{thm:tracialjoint}). Remarkably, the measure is fixed by $(A,B)$ and does not depend on $f$. Therefore, it can translate between functions of the matrix expression $xA+yB$ and functions of the scalar expression $xa+yb$.

Despite not being of the form of~\cref{eq:intro_master}, expressions such as relative entropy can also be written using this measure (see~\cref{thm:relativ_int}):
\begin{equation}
    D(A\|B)=2\int_{\mathbb R^2}a\log\frac ab\,\di\mu_{A,B}(a,b) \, .
    \label{eq:intro_relative}
\end{equation}
This allows us to find a variety of integral representations of the relative entropy just from simple calculus. 

In addition, it lets us relate the quantum relative entropy of two states to the classical relative entropy of two probability distributions.
Together with 
$\tr A=2\int a\,\di\mu_{A,B}$ and $\tr B=2\int b\,\di\mu_{A,B}$, one can define distributions $\di P_{A,B}=2a\,\di\mu_{A,B}$ and $\di Q_{A,B}=2b\,\di\mu_{A,B}$
that fulfill $D(A\|B) = D(P_{A,B},Q_{A,B})$.

These distributions have exactly the same binary-testing region as the quantum pair $(\rho,\sigma)$. 
The shape of this testing region keeps some of the noncommutative information of $\rho$ and $\sigma$ and we show that it is a polygon exactly when $\rho$ and $\sigma$ commute. 
This is equivalent to the testing curve $t\mapsto\tr[A-tB]_+$ being piecewise affine, answering an open question in~\cite{hiai2026hockeystick}.

The central mechanism of this paper is to turn scalar identities into trace identities and we summarize our main tools in~\cref{tab:tracial-joint-measure-dictionary}. 

\begin{table}[htpb]
\centering
\renewcommand{\arraystretch}{1.35}
\begin{tabular}{p{0.34\textwidth} p{0.40\textwidth} l}
\toprule
\textbf{Trace expression} & \textbf{Scalar quantity} & \textbf{Reference} \\
\midrule

$\tr[xA+yB]$ & $\displaystyle 2\int (xa+yb)\, \di\mu_{A,B}$ & \cref{thm:tracialjoint} \\

$\tr\bigl[(xA+yB)_{\pm}\bigr]$ & $\displaystyle 2\int (xa+yb)_{\pm}\, \di\mu_{A,B}$ & \cref{prop:basic} \\

$D(A\|B)$ & $\displaystyle 2\int a\log\frac{a}{b}\, \di\mu_{A,B}$ & \cref{thm:relativ_int} \\

$\tr\bigl[A\, K_s[B]\bigr]$ & $\displaystyle \int a\, \ee^{-sb}\, \di\mu_{A,B}$ & \cref{thm:diffdict} \\

$\tr\bigl[V\, \D\log(S)[W]\bigr]$ & $\displaystyle 2\int \frac{(ua+vb)(pa+qb)}{xa+yb}\, \di\mu_{A,B}$ & \cref{thm:diffdict} \\

$\tr\bigl[(uA+vB)\, \id_{\{A<sB\}}\bigr]$ & $\displaystyle 2\int  (ua+vb)\,\id_{\{a<sb\}}\, \di\mu_{A,B}$ & \cshref{thm:diffdict} \\
\bottomrule
\end{tabular}
\caption{
Scalar--operator dictionary with $S=xA+yB$, $V=uA+vB$,
$W=pA+qB$, and $K_s[X]=\int_0^1(1-r)\ee^{-rsX}\di r$. The precise 
assumptions are stated with the corresponding results.}
\label{tab:tracial-joint-measure-dictionary}
\end{table}

Most integral representations of the relative entropy presented here are known in equivalent form (see \cite{Frenkel_2023,Hirche_2024,liu2025layercakerepresentationsquantum}). 
Our contribution is mostly to lay the groundwork for a general application of Heinävaara's tracial joint measure for usage in quantum information.
To illustrate its usefulness, we further derive inequalities for relative entropy variance and the loss of relative entropy under positive trace-preserving maps. 

The paper is organized as follows.
\Cref{sec:tjsm} introduces $\mu_{A,B}$ and the first two equations of~\cref{tab:tracial-joint-measure-dictionary}. \Cref{sec:relent} proves
\cref{eq:intro_relative}, its classical reformulation and adds the last three entries of~\cref{tab:tracial-joint-measure-dictionary}. \Cref{sec:intrep} derives various integral representations, and \cref{sec:lowerbounds} lifts scalar
inequalities to tracial inequalities. Finally, \Cref{sec:testing} develops the binary-testing
interpretation, the curvature criterion, and the two trace inequalities. 
Further proof details are found in~\cref{sec:appdiff,sec:appscalar,sec:appderiv,sec:apptesting}.

\section{The tracial joint spectral measure}
\label{sec:tjsm}
\subsection{Preliminaries}
Throughout this work, we restrict ourselves to finite-dimensional Hilbert spaces  and use $\log(x)$ as the natural logarithm of $x$. Furthermore, we denote the functions
$x_+ = \max\{x,0\}$ and $x_- = \max\{-x,0\}$,
so that $x=x_+-x_-$, $|x|=x_++x_-$ and $x_\pm = \tfrac12(|x|\pm x)$. For a Hermitian matrix $X$
we write $X_\pm$ for its positive and negative part and often abbreviate $\tr[X_\pm]$ as $\tr[X]_\pm$.
In addition, we set $\id_{\{X<0\}}$ to be the spectral projection of $X$ onto $(-\infty,0)$ and use the notation
$\id_{\{X<Y\}} \coloneqq \id_{\{X-Y<0\}}$ for Hermitian $X,Y$. Finally, we use $\norm{\cdot}$ for the operator norm and
$\norm{\cdot}_1$ for the trace norm.

We begin by recalling Heinävaara's tracial joint measure.
\begin{theorem}[{\cite[Theorem 1.4]{Heinavaara2023}}]\label{thm:tracialjoint}
Let $A,B\in M_n(\mathbb C)$ be Hermitian. There is a positive measure $\mu_{A,B}$ on
$\mathbb R^2$ with $\mu_{A,B}[(0,0)]=0$ and the following property. If $f:\mathbb R\to\mathbb R$ is measurable, $f(0)=0$ and
\begin{equation}
    \int_{-M}^{M}\left|\frac{f(t)}{t}\right|\di t<\infty
    \qquad\text{for every }M>0,
\end{equation}
then, with
\begin{equation}
    H(f)(z)=\int_0^1\frac{1-t}{t}f(zt)\di t,
\end{equation}
one has
\begin{equation}
    \tr\bigl[H(f)(xA+yB)\bigr]
    =\int_{\mathbb R^2}f(xa+yb)\,\di\mu_{A,B}(a,b)
\end{equation}
for every $x,y\in\mathbb R$.
\end{theorem}

The measure does not depend on $f$ and is uniquely determined away from the
origin~\cite[Proposition~3.9]{Heinavaara2023}. $\mu_{A,B}[(0,0)]=0$ fixes it throughout this paper. When $f$ is continous and $H(f)$ twice differentiable, the scalar function is recovered from the trace term by
\begin{equation}\label{eq:inverseH}
    f(z)=2z(H(f))'(z)+z^2(H(f))''(z).
\end{equation}

\subsection{Properties of the measure}
\label{sec:props}
Heinävaara's measure can be shown to have the following properties.

\begin{proposition}[Properties of $\mu_{A,B}$]\label{prop:basic}

    Let $A,B\in M_n(\mathbb{C})$ be Hermitian and $x,y\in\mathbb{R}$. Then following propositions hold:
    \begin{enumerate}[(i)]

        \item For every $\alpha>0$,
        \begin{equation}\label{eq:basicmoments}
            \int_{\mathbb{R}^2} |xa+yb|^\alpha \meas = \frac{ \tr|xA+yB|^\alpha}{\alpha(\alpha+1)}  \, ,
        \end{equation}
        and the same identity holds with $|\cdot|$ replaced by $(\cdot)_+$, or
        $(\cdot)_-$.
        \item  The measure is supported in
        \begin{equation}\label{eq:basicsupp}
            \supp \, \, \mu_{A,B} \subseteq [-\norm{A},\norm{A}] \times [-\norm{B},\norm{B}] \, .
        \end{equation}
        Moreover,  $xA+yB\geq0$ implies $xa + yb \geq 0$ on its support.   
    \end{enumerate}
\end{proposition}
\begin{proof}
    \begin{enumerate}[(i)]
        \item Choosing $f(t)=|t|^\alpha$, we get
            \begin{equation}
                H(f)(z)=|z|^\alpha\int_0^1(1-t)t^{\alpha-1}\di t
                =\frac{|z|^\alpha}{\alpha(\alpha+1)} \,,
            \end{equation}
            which by~\cref{thm:tracialjoint} proves the statement. Taking $f(t)=(t_\pm)^\alpha$ gives the other cases.
        \item If $xA + yB \geq0$, the identity in (i) gives
        \begin{equation}
            2 \int(xa+yb)_- \meas = \tr[xA+yB]_- = 0 \, ,
        \end{equation}
        which shows the support restriction. For the more general bound, we apply~\cref{thm:tracialjoint} with $f(t) = \mathbf{1}_{\{|t|>M\}}$ first to $A$, then to $B$, and let $M$ go to the respective norms. This implies in a similar fashion the bounds on the support of $\mu_{A,B}$. \qedhere
    \end{enumerate}
\end{proof}

\begin{remark}\label{rem:bound}
    For $A,B\geq0$, every comparison $mB\leq A \leq MB$ transfers to
    \begin{equation}
        mb \leq a \leq Mb \qquad \mu_{A,B}\text{ - almost everywhere.}
    \end{equation}
    In particular, if $A,B>0$, then $a,b>0$ almost everywhere and $a/b$ is bounded above and away from zero.
\end{remark}

In the classical setting, the measure has a particularly simple form.

\begin{example}[Commuting case]\label{lem:commuting}
    Let $A,B\in M_n(\mathbb{C})$ be commuting Hermitian matrices with
    $A=\sum_{i=1}^n\lambda_i\proj{i}$ and $B=\sum_{i=1}^n\nu_i\proj{i}$ . Then
    \begin{equation}\label{eq:commutingmeas}
        \mu_{A,B} = \sum_{\substack{i=1\\(\lambda_i,\nu_i)\ne(0,0)}}^n \int_0^1 \frac{1-t}{t}\, \delta_{(t\lambda_i,\, t\nu_i)} \, \di t \, .
    \end{equation}
    In particular $\mu_{A,B}$ is supported on the $n$ line segments
    $\{(t\lambda_i,t\nu_i) : t\in(0,1]\}$ and is singular with respect to Lebesgue measure on
    $\mathbb R^2$.
\end{example}
\begin{proof}
    Since $A$ and $B$ are simultaneously diagonal, so is $xA+yB$, with eigenvalues
    $x\lambda_i+y\nu_i$. Hence, for every admissible $f$,
    \begin{align}
        \tr\bigl[H(f)(xA+yB)\bigr]
        = \sum_{i=1}^n H(f)(x\lambda_i+y\nu_i)
        = \sum_{i=1}^n \int_0^1 \frac{1-t}{t}\, f\bigl(t(x\lambda_i+y\nu_i)\bigr) \di t \, .
    \end{align}
    Because $t(x\lambda_i+y\nu_i) = x(t\lambda_i) + y(t\nu_i)$, the right-hand side is precisely
    $\int_{\mathbb{R}^2} f(xa+yb)$ integrated against the measure on the right-hand side
    of~\cref{eq:commutingmeas}. That measure therefore satisfies the defining property
    of~\cref{thm:tracialjoint}, and since the latter determines the measure uniquely away from the
    origin~\cite[Proposition~3.9]{Heinavaara2023}, the two agree.
\end{proof}

The first two entries of~\cref{tab:tracial-joint-measure-dictionary} follow immediately from~\cref{thm:tracialjoint} and~\cref{prop:basic}.

\section{Relative entropy via tracial joint spectral measure}
\label{sec:relent}
We use the following convention for the relative entropy of positive semidefinite $A$ and $B$
\begin{equation}\label{eq:convention}
    D(A\|B) \coloneqq \tr[A\log A] - \tr[A\log B] \quad \text{if } \supp \, A\subseteq\supp \, B ,
    \qquad D(A\|B) \coloneqq +\infty \quad\text{otherwise,}
\end{equation}
with $0\log0 = 0$.

In order to express the term $\tr[A\log A]$ as an integral, we consider the following special case of~\cref{thm:tracialjoint}.

\begin{corollary}\label{cor:xlogx}
    Let $A,B\geq0$ and let $f(t) = t(2\log t+3)$ for $t>0$ and $f(t)=0$ for $t\leq0$. Then, for all $x,y\geq0$,
    \begin{equation} \label{eq:xlogx}
        \tr[ (xA+yB)\log(xA+yB)] = \int_{\mathbb{R}^2} (xa+yb)\bigl(2\log(xa+yb)+3\bigr) \meas \, .
    \end{equation}
    In particular this gives
    \begin{equation}
        \tr[A\log A] =  \int_{\mathbb{R}^2} a (2\log a+3) \meas 
    \end{equation}
\end{corollary}
\begin{proof}
    The function $f$ is permitted by~\cref{thm:tracialjoint} and
    \begin{equation}
        H(f)(z)
        =z\int_0^1(1-t)(2\log z+2\log t+3)\di t
        =z\log z\qquad(z>0),
    \end{equation}
    since $\int_0^1(1-t)\log t\di t=-3/4$. Both sides vanish for $z\leq0$. 
\end{proof}
The following Lemma takes care of $\tr[A\log B]$.
\begin{lemma}\label{lem:alogb}
    Let $A,B\geq0$. If $\supp \, A\subseteq \supp \, B$, then
    \begin{equation}
    \tr[ A\log B] = \int_{\mathbb{R}^2} a(2\log b+3)\, \di\mu_{A,B}(a,b) \, .
    \end{equation}
\end{lemma}
\begin{proof}
    If $\supp \, A\subseteq \supp \, B$,  we can restrict both operators to $\supp \, B$ while leaving both the trace and the measure unchanged. This is due to the uniqueness of $\mu_{A,B}$ away from the origin (see \cref{thm:tracialjoint}).
    We may therefore assume $B>0$.

    Choose $f(t)=t(2\log t+3)$ as in~\cref{cor:xlogx}. Differentiating this identity at $(x,y)=(0,1)$
    gives on the trace side $\tr[A\log B]+\tr A$, and on the scalar side
    $\int a(2\log b+5)\meas$. Here, integral and derivative can be exchanged by dominated convergence.
\end{proof}

\begin{theorem}[Relative entropy via tracial joint spectral measure]\label{thm:relativ_int}
    Let $A,B\geq0$. If $\supp \, A\subseteq\supp \, B$, then
    \begin{equation}\label{eq:relativ_int}
        D(A\|B) = 2\int_{\mathbb{R}^2} a(\log a-\log b)\, \di \mu_{A,B}(a,b) \, .
    \end{equation}
\end{theorem}
\begin{proof}
    By~\cref{cor:xlogx} we have $\tr[A\log A] = \int_{\mathbb{R}^2} a(2\log a+3) \meas$ and by~\cref{lem:alogb}, \\$\tr[A\log B] = \int a(2\log b+3)\meas$.
    Subtracting them, the additive constants cancel and
    \begin{align}
        D(A\|B)
        &= \tr[A\log A]-\tr[A\log B]  \\
        &= \int_{\mathbb{R}^2} \big(a(2\log a+3)-a(2\log b+3)\big)\, \di \mu_{A,B}(a,b) \\
        &= 2 \int_{\mathbb{R}^2} a(\log a- \log b)\, \di \mu_{A,B}(a,b)  \, .
    \end{align}
\end{proof}

\begin{example}\label{rem:commuting_relative}
If $A=\sum_i\lambda_i\proj{i}$ and $B=\sum_i\nu_i\proj{i}$ commute, then substituting
\cref{eq:commutingmeas} into~\cref{eq:relativ_int} gives
\begin{equation}
    D(A\|B)
    =2\sum_i\int_0^1(1-t)\lambda_i\log\frac{t\lambda_i}{t\nu_i}\di t
    =\sum_i\lambda_i\log\frac{\lambda_i}{\nu_i}.
\end{equation}
Thus, the representation reduces to the classical formula. 
\end{example}

Because $\mu_{A,B}$ is a positive measure, the representation can also be seen as an entirely classical relative entropy.

\begin{corollary}[Classical relative entropy]\label{cor:classical}
    For $A,B \geq0$ with  $\supp \, A\subseteq\supp \, B$, define the positive measures
    \begin{equation}
        \di P(a,b) = 2 a \meas , \qquad \di Q(a,b) = 2 b \meas ,
    \end{equation}
    on $\mathbb{R}^2$. Then $P(\mathbb{R}^2)= \tr A$, $Q(\mathbb{R}^2)= \tr B$, and
\begin{equation}
    D(A\|B) = \int \log \frac{\di P}{\di Q} \, \di P = D(P\|Q) ,
\end{equation}
with the convention $D(P\|Q) = + \infty$ if $P \nll Q$. If $A$ and $B$ are
density matrices then $P$ and $Q$ are probability measures.
\end{corollary}
\begin{proof}
    By~\cref{prop:basic} (i), we have $\int \di P = \tr[A]$ and $\int \di Q = \tr[B]$. On the set where $a,b>0$ we have
    $\frac{\di P}{\di Q}(a,b) = \frac ab$, so $\int\log\frac{\di P}{\di Q}\di P = 2\int a\log\frac ab\meas$,
    which is $D(A\|B)$ by~\cref{thm:relativ_int}. 
\end{proof}

\subsection{Further integral expressions}
\label{sec:diffdict}
In the following, we generalize the trick used in the previous sections.
Expressions such
as $\tr[(uA+vB\,\id_{\{xA+yB<0\}}]$ or  $\tr[(uA+vB)\,\D\log(xA+yB)[pA+qB]]$ require an additional step to translate them into scalar integrals. For smooth
functions we differentiate~\cref{thm:tracialjoint} in the scalar parameters. For cutoffs, we instead compare the one-sided derivatives. This gives the last entries of~\cref{tab:tracial-joint-measure-dictionary}

Here $\D\log(X)[Y]=\frac{\di}{\di\tau}\log(X+\tau Y)\bigr|_{\tau=0}$ is the Fréchet derivative of the
matrix logarithm.

\begin{theorem}[Tracial joint measure expressions]\label{thm:diffdict}
    Let $A,B\geq 0$, then the following expressions hold:
    \begin{enumerate}[(i)]
        \item Set
        $S=xA+yB$, $V=uA+vB$, and $W=pA+qB$, and assume $S>0$. Then
        \begin{equation}\label{eq:diffdict_log}
            \tr\bigl[V\,\D\log(S)[W]\bigr]
            =2\int_{\mathbb R^2}\frac{(ua+vb)(pa+qb)}{xa+yb}\meas \, .
        \end{equation}
        \item For every $s \in \mathbb{R}$, we have
        \begin{equation}\label{eq:diffdict_ind}
            \tr[(xA+yB)\id_{\{A>sB\}}]
            =2\int_{\mathbb R^2}(xa+yb)\mathbf{1}_{\{a>sb\}}\meas,
        \end{equation}
        and the same identity holds with both inequalities replaced by $\geq$ or reversed.

        \item  For $s>0$ and
        $K_s[X]=\int_0^1(1-r)\ee^{-rsX}\di r$, we have
        \begin{equation}\label{eq:diffdict_exp}
            \int_{\mathbb R^2}a\ee^{-sb}\meas=\tr[A K_s[B]],
            \qquad
            \int_{\mathbb R^2}a\ee^{-sa}\meas=\tr[A K_s[A]].
        \end{equation}
    \end{enumerate}
\end{theorem}
\begin{proof}
The theorem is proven in~\cref{sec:appdiff}.
\end{proof}

\section{Integral representations}
\label{sec:intrep}
We now derive several integral representations of $D(A\|B)$. 
First we write $a\log(a/b)$ as a scalar integral, then we integrate the expression against
$\mu_{A,B}$, interchange the integrals and use the dictionary~\cref{tab:tracial-joint-measure-dictionary}. 
Here is a selection of integral identities:
\begin{lemma}[Integral expressions]\label{lem:scalar-integrals}
    Let $A,B\geq0$ with $\supp \, A\subseteq\supp \, B$, then for $\mu_{A,B}$-almost everywhere, we have
    \begin{align}
        a \log \frac{a}{b} &= a-b+\int_0^1 \frac{(a-sb)_-}{s}\di s + \int_1^\infty \frac{(a-sb)_+}{s} \di s \, , \label{eq:rep1}\\
        a \log \frac{a}{b} &= a-b+\int_1^\infty \frac{(a-sb)_+}{s} \di s + \int_1^\infty \frac{(b-sa)_+}{s^2} \di s \, ,\label{eq:rep2}\\
        a \log \frac{a}{b} &= a-b+\int_{-\infty}^\infty \frac{\bigl((1-s)a+sb\bigr)_-}{|s|(s-1)^2} \di s \, , \label{eq:rep3}\\
        a \log \frac{a}{b} &= a-b+\int_0^1 \frac{(1-s)(a-b)^2}{sa+(1-s)b} \, \di s \, , \label{eq:rep4}\\
        a \log \frac{a}{b} &= b \int_0^{\infty} (1+\log s)\, \mathbf{1}_{\{a>sb\}}\, \di s \, , \label{eq:rep5}\\
        a \log \frac{a}{b} &=  a\int_0^\infty \frac{\ee^{-sb}-\ee^{-sa}}{s} \,\di s\, . \label{eq:rep6}
    \end{align}
\end{lemma}
\begin{proof}
    See~\cref{sec:appscalar}.
\end{proof}

Interchanging the integrals is made precise by Fubini-Tonelli. After that, it is just an application of the entries of~\cref{tab:tracial-joint-measure-dictionary}.
\begin{corollary}\label{cor:integral-reps}
    Let $A,B\geq0$ with $\supp \, A\subseteq\supp \, B$, then we have
    \begin{align}
        D(A\|B) &= \tr[A-B]+\int_0^1 \frac{\tr[A-sB]_-}{s} \, \di s + \int_1^\infty \frac{\tr[A-sB]_+}{s} \, \di s \, , \label{eq:int_rep1} \\
        D(A\|B) &= \tr[A-B] + \int_1^\infty \frac{\tr[A-sB]_+}{s}\, \di s + \int_1^\infty \frac{\tr[B-sA]_+}{s^2}\, \di s \, , \label{eq:int_rep2}\\
        D(A\|B) &= \tr[A-B] + \int_\mathbb{R} \frac{\tr[(1-s)A+sB]_-}{|s|(s-1)^2} \, \di s \, , \label{eq:int_rep3} \\
        D(A\|B) &= \tr[A-B] + \int_0^1 (1-s)\, \tr[ (A-B)\, \D\log(sA + (1-s)B)[A-B]] \, \di s \, ,\label{eq:int_rep4}\\
        D(A\|B) &= \int_0^\infty (1 + \log s)\,
        \tr[B\, \id_{\{A>sB\}}] \, \di s \, , \label{eq:int_rep5}\\
        D(A\|B) &= 2\int_0^\infty \frac{\tr[A (K_s[B]-K_s[A])]}{s} \, \di s \, , \label{eq:int_rep6}
    \end{align}
    where
    \begin{equation}
        K_s[X] = \int_0^1 (1-r) \ee^{-rsX} \, \di r \, .
    \end{equation}
\end{corollary}
\begin{proof}
    By~\cref{thm:relativ_int}, we have $D(A\|B)=2\int a \log\frac{a}{b} \, \di \mu$. Then we can use~\cref{lem:scalar-integrals} to rewrite $a\log \frac{a}{b}$. To justify the use of Fubini-Tonelli, we have to show absolute-integrability. Since we assume $\supp \, A\subseteq\supp \, B$, $D(A\|B)$ is finite. All expressions with positive integrand therefore automatically satisfy the requirement. It remains to show the same for integrands that can switch sign. This is also the case since we have on the support of $\mu_{A,B}$
    \begin{equation}
        \int_0^\infty |b+b\log s|\, \mathbf{1}_{\{a>sb\}}\, \di s \leq 2b + a\left|\log\frac{a}{b} \right|
    \end{equation}
    and
    \begin{equation}
        \int_0^\infty \left|a\frac{\ee^{-sb}-\ee^{-sa}}{s} \right|\,\di s = a\left|\log\frac{a}{b} \right|
    \end{equation}
    which together with~\cref{tab:tracial-joint-measure-dictionary} shows absolute-integrability.
\end{proof}

Of the six representations, \cref{eq:int_rep1,eq:int_rep2,eq:int_rep3,eq:int_rep5} can be found in similar or equivalent form in~\cite{Frenkel_2023,Hirche_2024,liu2025layercakerepresentationsquantum}.

\begin{remark}
    If $A=\rho$ and $B=\sigma$ are density matrices, then \cref{eq:int_rep1} can be simplified to
    \begin{equation}
        D(\rho \|\sigma) = \int_0^\infty \frac{\norm{\rho-s\sigma}_1-|1-s|}{2s} \,  \di s \, .
    \end{equation}
    In this form, monotonicity under positive trace-preserving maps is immediate from monotonicity of the trace norm.
\end{remark}

\begin{remark}
    Choosing $A=\rho_{ABC}$ and $B=\exp(\log \rho_{AB} + \log \rho_{BC} - \log \rho_B)$, we have $D(A\|B) = I(A:C|B)$.
    Using Lieb's three matrix Golden Thompson inequality, we know that $\tr B \leq 1$.  Therefore, all these expressions give us stronger versions of strong subadditivity. In the case of $I(A:C|B)=0$, the remainder terms vanish and give structural constraints on the state $\rho$, c.f.~\cite{Hayden_2004}.
\end{remark}

\section{Lower bounds on the relative entropy}
\label{sec:lowerbounds}
The same mechanism from the previous chapter turns scalar inequalities into bounds on the relative entropy. 

\begin{lemma}[Scalar lower bound]\label{lem:lowerbound}
    For $a\geq0,b>0$, we have
    \begin{equation}
        a \log \frac{a}{b} \geq a-b + \frac{3(a-b)^2}{2(a+2b)}
    \end{equation}
\end{lemma}
\begin{proof}
    By~\cref{eq:rep4}, we have $a \log \frac{a}{b} = a-b+(a-b)^2\int_0^1 1/(sa+(1-s)b) \, (1-s) \di s$. Treating $2(1-s)$ as a probability distribution on $[0,1]$, we can apply Jensen's inequality to the function $x\to1/x$. Since its mean is $1/3$, this gives $\int_0^1\frac{1-s}{sa+(1-s)b} \di s \geq \frac{3}{2(a+2b)}$.
\end{proof}

From this inequality, Pinsker can be derived by a simple application of Cauchy-Schwarz. In a similar fashion, one can use Hölder to derive a family of lower bounds.

\begin{corollary}[Lower bounds]\label{cor:pinsker}
    Let $A,B\geq0$. Then
    \begin{equation}\label{eq:pinsker}
        D(A\|B) \geq \tr[A-B] + \frac{3}{2}\frac{\norm{A-B}_1^2}{\tr[A+2B]} \, .
    \end{equation}
    For density matrices $\rho,\sigma$ it reduces to Pinsker's inequality 
    \begin{equation}\label{eq:pinskerstates}
        D(\rho\|\sigma) \geq \frac{1}{2}\norm{\rho-\sigma}_1^2 \, .
    \end{equation}
    Using Hölder, one obtains for all $p>1$
    \begin{equation}\label{eq:holderpinskerD}
        D(A\|B)\geq\tr[A-B]+3c_p\, \frac{\bigl(\tr|A-B|^{2/p}\bigr)^p}{\bigl(\tr[(A+2B)^{1/(p-1)}]\bigr)^{p-1}}.
    \end{equation}
    with
    \begin{equation}
        c_p=\left(\frac{p^2}{2(2+p)}\right)^p\frac{p^{p-1}}{(p-1)^{2(p-1)}}.
    \end{equation}
\end{corollary}
\begin{proof}
    Define $\delta=a-b$ and $m=a+2b$, then Cauchy--Schwarz allows us to bound $\int_{\mathbb{R}^2} \frac{(a-b)^2}{a+2b} \meas$ via
    \begin{equation}
        \int\frac{\delta^2}{m}\meas 
        \geq  \frac{\left(\int|\delta|\meas\right)^2}{\left(\int m\meas\right)}
    \end{equation}
    \cref{tab:tracial-joint-measure-dictionary} then lifts the expression to~\cref{eq:pinsker}. Analogously, applying  Hölder's inequality gives
    \begin{equation}\label{eq:holderrearranged}
        \int\frac{|\delta|^2}{m}\meas
        \geq
        \frac{\bigl(\int|\delta|^{2/p}\meas\bigr)^p}
        {\bigl(\int m^{1/(p-1)}\meas\bigr)^{p-1}}.
    \end{equation}
    Applying~\cref{prop:basic} (i) then leads to the concrete $c_p$.
\end{proof}

\begin{remark}[Log-sum inequality]\label{rem:logsum}
    Applying the usual log-sum inequality to the measure $P,Q$ defined in~\cref{cor:classical} gives
    \begin{equation}
        D(A\|B) \geq \tr A \, \log \frac{\tr A}{\tr B}
    \end{equation}
\end{remark}

\section{Binary testing and further consequences }
\label{sec:testing}

\Cref{cor:classical} shows that the classical pair $(P,Q)$ reproduces
the relative entropy of $(\rho,\sigma)$. We now show that it produces the full binary-testing geometry as well. 
In fact, it connects to the recent analysis done in~\cite{gour2026noncommutativity}. The there discussed layer-cake Stieltjes measure (see also~\cite{liu2025layercakerepresentationsquantum}) arises as the push-forward of the tracial joint measure. 

Throughout this section, $\rho$ and $\sigma$ are states with
$\supp \,\rho\subseteq\supp \,\sigma$. If necessary, we restrict ourselves to $\supp \,\sigma$ such that $\sigma>0$. Therefore we always have $P\ll Q$, and
\begin{equation}\label{eq:likelihoodratio}
    R(a,b)\coloneqq \frac ab=\frac{\di P}{\di Q}(a,b) <\infty
    \qquad Q\text{-almost everywhere}.
\end{equation}
Let $\nu=R_\#Q$ be the law of this likelihood ratio. It is easy to check that it is a probability measure on
$[0,\infty)$ with mean one. In particular, we can formulate the relative entropy in terms of this measure. \cref{thm:relativ_int} implies $D(\rho\|\sigma) = \int_0^\infty r \log r \, \di \nu$. In fact, we can define a whole family of expressions by \begin{equation}
    D_f^{TJ}(\rho\|\sigma) \coloneqq D_f(P\|Q) =  \int_0^\infty f(r) \, \di \nu(r) \, .
\end{equation}

\begin{remark}[Layer cake $f$-divergences]\label{rem:layercakef}
For a finite continuous convex function $f:[0,\infty)\to\mathbb R$ with $f(1)=0$, the tracial-measure defines the following $f$-divergence:
\begin{equation}
    D_f^{\mathrm{TJ}}(\rho\|\sigma) = D_f(P\|Q) = 2\int_{\mathbb{R}^2}b\,f\left(\frac {a}{b}\right)\,\di\mu_{\rho,\sigma}(a,b)\,.
\end{equation}
This definition coincides with the quantum layer cake $f$-divergence introduced in~\cite[Definition~1]{liu2025layercakerepresentationsquantum}. By \cref{thm:diffdict} (ii), we have 
\begin{equation}
    \tr [\sigma\id_{\{\rho>\gamma\sigma\}}]= 2\int_{\mathbb{R}^2}b\,\mathbf{1} _{\{a>\gamma b\}}\, \di\mu_{\rho,\sigma}(a,b) = Q\{R>\gamma\} = \nu[(\gamma,\infty)] \,.
\end{equation}
Using the scalar identity $f(r) = f(0) + \int_0^\infty \mathbf{1}_{\{\gamma<r\}} \, \di f(\gamma)$, we get 
\begin{equation}
    D_f^{TJ}(\rho\|\sigma) = \int_0^\infty f(r) \, \di \nu(r) = f(0)+\int_0^\infty \nu[(\gamma,\infty)] \, \di f(\gamma) = f(0) + \int_0^\infty \tr [\sigma\id_{\{\rho>\gamma\sigma\}}] \, \di f(\gamma) \, .
\end{equation}
In particular, \cite[Definition~2]{liu2025layercakerepresentationsquantum} can be written as
\begin{equation}
    Q_\alpha^{\mathrm{LC}}(\rho\|\sigma) \coloneqq\alpha\int_0^\infty\gamma^{\alpha-1} \tr\bigl[\sigma\id_{\{\rho>\gamma\sigma\}}\bigr]\di\gamma 
    = \int_0^\infty r^\alpha \, \di \nu(r)
    =\int\left(\frac{\di P}{\di Q}\right)^\alpha\di Q \,.
\end{equation}
Thus the layer cake R\'enyi divergence is precisely $D_\alpha(P\|Q)$ for $\alpha\in(0,1)\cup(1,\infty)$. 
\end{remark}

For a quantum pair and a classical pair, define the testing regions
\begin{align}
    \cT_{\rm q}(\rho,\sigma) &\coloneqq\left\{\bigl(\tr[T\rho],\tr[T\sigma]\bigr):0\leq T\leq\id\right\},\label{eq:Tq}\\
    \cT_{\rm c}(P,Q)
    &\coloneqq \left\{\left(\int\varphi\,\di P,\int\varphi\,\di Q\right):0\leq\varphi\leq 1, \text{ measurable}\right\}.
    \label{eq:Tc}
\end{align}
Both are compact convex subsets of $[0,1]^2$.

\begin{proposition}[Classical testing]\label{prop:testingregion}
For states $\rho,\sigma$ with $\supp \,\rho\subseteq\supp \,\sigma$, we have
\begin{equation}\label{eq:testingregion}
    \cT_{\rm q}(\rho,\sigma)=\cT_{\rm c}(P,Q).
\end{equation}
In particular, for every $t\geq0$,
\begin{equation}\label{eq:stoploss}
    \tr(\rho-t\sigma)_+ =\int_0^\infty(r-t)_+\,\di\nu(r) \, .
\end{equation}
\end{proposition}
\begin{proof}
    We are going to show that the two testing regions give the same optimal value for every weighted sum of the two testing probabilities. For this, consider $\tr[T(x\rho + y\sigma)]$. Maximizing the test $T$ then gives $\tr[x\rho+y\sigma]_+$. Analogously, we find for $\max_\varphi \left( x \int \varphi \, \di P + y \int \varphi \, \di Q\right)=2\int(xa+yb)_+ \, \di \mu$. By~\cref{prop:basic} (i), these two expressions are equal. The two regions therefore have the same maximum for every $(x,y)$. Since both regions are compact and convex, they have to be the same. 
    \cref{eq:stoploss} follows from $\tr[\rho-t\sigma]_+ = 2 \int(a-tb)_+\, \di \mu = \int (\frac{a}{b}-t)_+\, \di Q = \int (r-t)_+ \, \di \nu(r)$.
\end{proof}

Quantum testing regions and their Lorenz curves were studied before in~\cite{buscemi2017lorenz}, and their
description by a layer cake Stieltjes measure was made explicit in~\cite{gour2026noncommutativity}.
Here, we connect them to the tracial joint measure.
In particular, we can make the following observation.
\begin{corollary}[Data processing]\label{cor:data-processing}
Let $\Phi$ be a positive trace-preserving map and denote by $P', Q'$ and $\nu'$ the classical pair and likelihood measure corresponding to $(\Phi(\rho),\Phi(\sigma))$. There exists a Markov kernel $K$ such that
\begin{equation}\label{eq:blackwellkernel}
    KP=P',\qquad KQ=Q' \, .
\end{equation}
Moreover, there exists a coupling $(R,R')$ with laws $\nu,\nu'$ that satisfies $\mathbb{E}[R\mid R']=R'$.
Therefore we have $\nu \succeq_{\rm cx}\nu '$ and every $f$-divergence from~\cref{rem:layercakef} satisfies 
\begin{equation}
    D_f^{TJ}(\rho\|\sigma) \geq D_f^{TJ}(\Phi(\rho)\|\Phi(\sigma)) \,.
\end{equation}
\end{corollary}
\begin{proof}
The adjoint $\Phi^*$ is positive and unital, so every test $0\leq T\leq\id$ on the output pair induces a test $0\leq \Phi^*(T) \leq 1$ on the input with the same probabilites. Thus, we have $\cT_{\rm q}(\Phi(\rho),\Phi(\sigma))\subseteq\cT_{\rm q}(\rho,\sigma)$. 
By \cref{prop:testingregion}, the same inclusion holds for the two classical testing regions. The binary Blackwell theorem~\cite{blackwell1953} therefore gives a Markov kernel $K$ with $KP=P'$ and $KQ=Q'$.
We can construct the coupling by considering $(X,Y)$ under the joint law $Q(\di x)K(x,\di y)$ and defining $R = \frac{\di P}{\di Q}(X)$ and $R' = \frac{\di P'}{\di Q'}(Y)$. Since $X$ and $Y$ have laws $Q$ and $Q'$, the ratios $R,R'$ have laws $\nu$ and $\nu'$. Then, for every bounded measurable function $g$, we have 
\begin{align}
    \mathbb{E}[Rg(Y)] &= \int \int \frac{\di P}{\di Q}(x) g(y) K(x,\di y) Q(\di x) = \int \int g(y) K(x,\di y) P(\di x) = \int g(y) \di(KP)(y) \\ &= \int g(y) \di P'(\di y) = \int g(y) \frac{\di P'}{\di Q'}(y) Q'(\di y) = \mathbb{E}[R'g(Y)] \, .
\end{align}
Hence $\mathbb{E}[R|Y] = R'$. Since $R'$ is a function of $Y$, conditioning on it gives by the tower property
\begin{equation}
    \mathbb{E}[R|R'] = \mathbb{E}[\mathbb{E}[R|Y]|R'] = \mathbb{E}[R'|R']=R' \, .
\end{equation}
For a convex function $f$ one then gets by conditional Jensen that 
\begin{equation}
    f(R') = f(\mathbb{E}[R|R'])\leq\mathbb{E}[f(R)|R'].
\end{equation}
Taking the expectation then gives $\mathbb{E}[f(R')]\leq\mathbb{E}[f(R)]$ or equivalently $\int f \, \di \nu' \leq \int f \, \di \nu$, which implies the convex order between the measures.
\end{proof}

\subsection{Noncommutativity as curvature}
Let us define the testing curve $h(t)\coloneqq\tr(\rho-t\sigma)_+$ for $t>0$.  For commuting states this function is a finite sum of positive-part functions and piecewise affine. One could wonder, if the converse is true as well. Here we show that this is indeed the case and that the noncommutativity of two states is closely related to the curvature of $h(t)$ and the geometry of the testing region.

\begin{theorem}[Curvature detects noncommutativity]\label{thm:curvature}
For states $\rho$ and $\sigma$ with $\sigma>0$, the following conditions are equivalent.
\begin{enumerate}[(i)]
    \item $\rho$ and $\sigma$ commute
    \item $h(t) = \tr(\rho-t\sigma)_+$ has finitely many affine pieces
    \item $\nu$ has finite support
    \item $\cT_{\rm q}(\rho,\sigma)$ is a polygon
\end{enumerate}
More explicitly, suppose $0$ is not an eigenvalue of $\rho-t\sigma = \sum_i\lambda_i\proj{u_i}$. Then
\begin{equation}\label{eq:curvature}
    h''(t)
    =2\sum_{\lambda_i>0>\lambda_j}
      \frac{|\langle u_i,\sigma u_j\rangle|^2}{\lambda_i-\lambda_j}\geq0.
\end{equation}
\end{theorem}
\begin{proof}
    See~\cref{sec:apptesting}
\end{proof}
\cref{thm:curvature} gives the direct finite-dimensional argument conjectured in~\cite[Remark~VII.5]{hiai2026hockeystick} and relates to the polyhedral criteria for commutativity in~\cite{li2020jointnumerical}.

\subsection{Variance and contraction inequalities}
As an application, we derive the following trace inequalities.
For states, we define the tracial joint variance as
\begin{equation}\label{eq:VTJ}
    V_{\rm TJ}(\rho\|\sigma) \coloneqq \int_0^\infty r (\log r-D(\rho\|\sigma))^2\, \di \nu(r) \, .
\end{equation}
This is the classical relative-entropy variance of $(P,Q)$. The usual
quantum relative-entropy variance is given by
\begin{equation}\label{eq:Vquantum}
    V(\rho\|\sigma) \coloneqq \tr\! \left[\rho\bigl(\log\rho-\log\sigma-D(\rho\|\sigma)\bigr)^2\right].
\end{equation}

\begin{proposition}[Variance bound]\label{prop:variance}
For states $\rho,\sigma$, we have
\begin{equation}\label{eq:variancebound}
    V_{\rm TJ}(\rho\|\sigma)\leq V(\rho\|\sigma).
\end{equation}
\end{proposition}
\begin{proof}
Define the functions $f(\alpha)=\int r^\alpha\,\di\nu(r)$ and $g(\alpha)=\tr(\rho^\alpha\sigma^{1-\alpha})$.
Then $f'(1) = \int r \log r \, \di \nu(r) = D(\rho\|\sigma)=g'(1)$ and $f''(1) = \int r(\log r)^2 \, \di \nu (r)$.
For $\alpha \in(0,1)$, the comparison in~\cite[Proposition~3.9]{beigi2025properties} shows
$f(\alpha)\leq g(\alpha)$. 
At $\alpha=1$, we have $f(1)=g(1)=1$ and $f'(1)=g'(1)=D(\rho\|\sigma)$.
Thus $h\coloneqq g-f$ satisfies $h(\alpha)\geq0$ for $\alpha<1$ and
$h(1)=h'(1)=0$. Dividing by $(1-\alpha)^2$ and taking $\alpha\uparrow1$ gives
$0\leq \frac{1}{2} h''(1)$ and therefore $f''(1)\leq g''(1)$. 
Since $g''(1) = \tr[\rho(\log \rho - \log \sigma)^2]$. The bound follows from subtracting $D(\rho\|\sigma)^2$ on both sides.
\end{proof}

\begin{corollary}[One-shot likelihood tails]\label{cor:tails}
For states $\rho,\sigma$, set $d=D(\rho\|\sigma)$ and $r>0$. Then
\begin{equation}\label{eq:likelihoodtails}
\begin{split}
    \tr\!\left[\rho\,\id_{\{\rho<\ee^{d-r}\sigma\}}\right]
    &\leq\frac{V_{\rm TJ}(\rho\|\sigma)}{V_{\rm TJ}(\rho\|\sigma)+r^2}
    \leq\frac{V(\rho\|\sigma)}{V(\rho\|\sigma)+r^2},\\
    \tr\!\left[\rho\,\id_{\{\rho>\ee^{d+r}\sigma\}}\right]
    &\leq\frac{V_{\rm TJ}(\rho\|\sigma)}{V_{\rm TJ}(\rho\|\sigma)+r^2}
    \leq\frac{V(\rho\|\sigma)}{V(\rho\|\sigma)+r^2}.
\end{split}
\end{equation}
\end{corollary}
\begin{proof}
Under $P$, the random variable $X=\log R$ has mean $D(\rho\|\sigma)$ and variance $V_{\rm TJ}(\rho\|\sigma)$. Then~\cref{thm:diffdict} (ii) gives
    $\tr[\rho\,\id_{\{\rho<s\sigma\}}]=P\{R<s\}$ and 
    $\tr[\rho\,\id_{\{\rho>s\sigma\}}]=P\{R>s\}$.
Applying Cantelli's inequality to
the two tails of $X-D(\rho\|\sigma)$, followed by~\cref{prop:variance}, proves the statement.
\end{proof}

For another application of the tracial joint measure, define
\begin{equation}\label{eq:chiBKM}
    \chi^2(\rho\|\sigma)
    \coloneqq \tr\left[(\rho-\sigma)\,\D\log(\sigma)[\rho-\sigma]\right] = \int(r-1)^2 \, \di \nu \, .
\end{equation}
Then we can formulate the following estimate.
\begin{proposition}[Contraction bound]\label{thm:stability}
Let $\rho,\sigma$ be states and $\Phi$ a positive and trace preserving map.
Furthermore, suppose there exist positive finite $m$ and $M$ such that $m\sigma\leq\rho\leq M\sigma$.
Then define
\begin{align}
    \Delta D&\coloneqq D(\rho\|\sigma)-D(\Phi(\rho)\|\Phi(\sigma)),\label{eq:DeltaD}\\
    \Delta\chi^2&\coloneqq \chi^2(\rho\|\sigma) -\chi^2(\Phi(\rho)\|\Phi(\sigma)).\label{eq:Deltachi}
\end{align}
Then
\begin{equation}\label{eq:stability}
    \frac{\Delta\chi^2 }{2M}\leq\Delta D \leq\frac{\Delta\chi^2}{2m}\,.
\end{equation}
Moreover,
\begin{equation}\label{eq:hbound}
    \Delta D\geq\frac{1}{2M} \sup_{\gamma\geq0}\left(\tr(\rho-\gamma\sigma)_+ -\tr(\Phi(\rho)-\gamma\Phi(\sigma))_+ \right)^2.
\end{equation}
\end{proposition}
\begin{proof}
    Take $R,R'$ and the coupling as defined in~\cref{cor:data-processing}. By~\cref{rem:bound}, $R$ and $R'$ take values in $[m,M]$ and by~\cref{cor:data-processing}, we have $\mathbb{E}[R|R']=R'$. Note, that $\chi^2(\rho\|\sigma) = \mathbb{E}[(R-1)^2]$ and since $\mathbb{E}[R]=1$, this gives $\chi^2(\rho\|\sigma) = \mathbb{E}[R^2]-1$. Therefore we can write $\Delta \chi^2 = \mathbb{E}[R^2]-\mathbb{E}[(R')^2]=\mathbb{E}[(R-R')^2]$ where the last step uses the property of the coupling. By defining $f(r) = r\log r$, we can write $\Delta D= \mathbb{E}[f(R)] - \mathbb{E}[f(R')]$ as well.
    Now consider the functions $g_1(r)=f(r)-\frac{r^2}{2M}$ and $g_2(r)=\frac{r^2}{2m}-f(r)$. These are convex on $[m,M]$ since $f''(r)=\frac{1}{r}$ and by applying~\cref{cor:data-processing}, we get $\mathbb{E}[g_1(R)]\geq\mathbb{E}[g_1(R')]$ and  $\mathbb{E}[g_2(R)]\geq\mathbb{E}[g_2(R')]$. This translates into $\frac{\Delta \chi^2}{2M}\leq \Delta D$ and $\Delta D \leq \frac{\Delta\chi^2}{2m}$. For the final inequality, consider the function $g_\gamma(r) = (r-\gamma)_+$. By~\cref{eq:stoploss}, we have $\Delta h(\gamma) \coloneqq \tr[\rho-\gamma\sigma]_+ - \tr[\Phi(\rho) - \gamma \Phi(\sigma)]_+ = \mathbb{E}[g_\gamma(R)] - \mathbb{E}[g_\gamma(R')]$ which is positive by~\cref{cor:data-processing} due to $g_\gamma$ being convex. Since $g_\gamma$ also fulfills $|g_\gamma(x) - g_\gamma(y)|\leq|x-y|$, we find $\Delta h(\gamma) = \mathbb{E}[g_\gamma(R) - g_\gamma(R')]\leq \mathbb{E}|g_\gamma(R) - g_\gamma(R')|\leq \mathbb{E}|R-R'|$. Then applying Cauchy-Schwarz leads to $\Delta h \leq \sqrt{\mathbb{E}[(R-R')^2]}$ which by the previous argument  gives $\Delta \chi^2\geq (\Delta h)^2$ for every $\gamma \geq0$.
\end{proof}

\section{Conclusion}
The main point of our work is that Heinävaara's measure can be fruitfully used in quantum information theory and that it allows to derive trace (in)-equalities from scalar expressions. We use it in particular for integral representations of the relative entropy and characterization of the quantum binary-testing region, but the possible applications are not necessarily limited to that. 

An interesting direction could be for instance the application to continuity. One can bound the distance between two integral expressions by $\left| \int f(x) \, \di \mu_1 - \int f(x')\, \di \mu_2\right|\leq\inf_\Pi \int |f(x)-f(x')| \, \di \Pi$, where the optimization is done over couplings between $\mu_1$ and $\mu_2$. This is essentially a Wasserstein-distance between tracial joint measure. In the case of von Neumann entropy, this reproduces known bounds, but in the general case it remains elusive.

\paragraph{Acknowledgements:}
We thank Joel Tropp for giving a talk at ETH that made us aware of Heinävaara's thesis and David Sutter for insightful and helpful discussions. Furthermore, we acknowledge GPT 5.6 Sol for assistance in \cref{sec:testing} and revision of the document. The author takes responsibility for the final manuscript.
This work was supported by the ETH Quantum Center and SNSF-Grant No. 200021E\_232425.


\appendix
\section{Trace derivative}
\label{sec:appderiv}

\begin{lemma}\label{lem:trace-deriv}
Let $S,V\in M_n(\mathbb{C})$ be Hermitian, let $I\subseteq\R$ be an open interval, let
$f\in C^1(I)$, and suppose there is $\varepsilon>0$ such that $\spec(S+\tau V)\subset I$ for all
$|\tau|<\varepsilon$. Then $\tau\mapsto\tr[f(S+\tau V)]$ is differentiable on
$(-\varepsilon,\varepsilon)$ and
    \begin{equation}
        \frac{\di}{\di \tau} \tr[f(S+\tau V)] = \tr[f'(S+\tau V)\,V] ,
    \end{equation}
in particular at $\tau=0$ it equals $\tr[f'(S)V]$.
\end{lemma}
\begin{proof}
Follows from the divided-difference
representation of the Fréchet derivative of a matrix function
\cite[Chapter~3]{Higham2008Functions}, whose diagonal entries in an eigenbasis of $S+\tau V$ are
$f'$ evaluated at the eigenvalues.
\end{proof}

\section{Differentiation under the tracial joint spectral measure}
\label{sec:appdiff}
Here, we prove the different items of~\cref{thm:diffdict}.
Recall that $\D\log(S)[\cdot]$ denotes the Fréchet derivative\footnote{Technically, this is the Gateaux derivative.} of the matrix logarithm at $S$.
\begin{equation}
    \D\log(S)[W] = \frac{\di}{\di\tau}\log(S+\tau W)\Bigr|_{\tau=0} \, .
\end{equation}

\begin{lemma}[i]\label{lem:logderiv}
Let $A,B\geq0$, $x,y,u,v,p,q\in\mathbb{R}$ and define $S = xA+yB$ , $V = uA+vB $, and $W = pA+qB$ .
Furthermore assume $S>0$. Then
\begin{equation}\label{eq:logderiv}
    \tr\left[V\, \D\log(S)[W]\right] = 2\int_{\mathbb{R}^2} \frac{(ua+vb)(pa+qb)}{xa+yb} \meas .
\end{equation}
\end{lemma}
\begin{proof}
     By~\cref{cor:xlogx}, we have
     \begin{equation}
         \tr[S\log S] = \int z(2 \log z+3) \, \di \mu, \qquad z=xa+yb\,.
     \end{equation}
     If we now replace $S$ by $S+tV+rW$ and take the mixed derivative, we get
     \begin{equation}
         \partial_t\partial_r \tr[(S+tV+rW) \log(S + tV+rW)]\Bigr|_{t=r=0} = \tr[V\, \D\log(S)[W]] \, .
     \end{equation}
     On the scalar side, by dominated convergence, we can take the derivative into the integral, which leads to
     \begin{equation}
         \tr[V\, \D\log(S)[W]] = 2 \int\frac{(ua+vb)(pa+qb)}{xa+yb} \meas \, .
     \end{equation}
\end{proof}

\begin{lemma}[ii]\label{lem:proj}
    Let $A,B\geq0$ and $x,y \in \mathbb{R}$, then for every $s\in\mathbb{R}$, we have
    \begin{equation}
        \tr[(xA+yB)\id_{\{A>sB\}}] = 2 \int(xa+yb)\mathbf{1}_{\{a>sb\}} \meas
    \end{equation}
    and
    \begin{equation}
        \tr[(xA+yB)\id_{\{A<sB\}}] = 2 \int(xa+yb)\mathbf{1}_{\{a<sb\}} \meas
    \end{equation}
\end{lemma}
\begin{proof}
    By~\cref{prop:basic}, we have 
        $\tr[A-sB]_+ = 2 \int (a-sb)_+ \meas.$
    Since $b\geq0$ and $\int b \, \di \mu < \infty$, its right derivative on the scalar side is $-2\int b \mathbf{1}_{\{a>sb\}} \, \di \mu$. On the matrix side, one gets
        $-\tr[B \id_{\{A>sB\}}] $
    for the right derivative. 
    This can be seen by using $\tr[X]_+ = \max_{0\leq T\leq\id}\tr[TX]$. Maximizers are of the form $T=\id_{\{X>0\}} + K$, where $\id_{\{X>0\}}$ is the projector onto the strict positive part of $X$ and $K$ a projector onto $\ker X$. The usual derivative would lead to $-\tr[(\id_{\{X>0\}} + K)B]$, which is not defined. But since we are taking the right derivative, (which is the largest of all these terms), we obtain  $-\tr[B \id_{\{A>sB\}}]$. Similarly the left derivative gives, $-\tr[B \id_{\{A\geq sB\}}] = -2\int b \mathbf{1}_{\{a\geq sb\}} \, \di \mu$. 
    To obtain the expression with $xA +yB$, note that $xA+yB = x(A-sB)+(xs+y)B$. Then
    \begin{align}
        \tr[(xA+yB)\id_{\{A>sB\}}] &= x\tr[(A-sB)\id_{\{A>sB\}}]+(xs+y)\tr[B\id_{\{A>sB\}}] \\
        &= x \tr[A-sB]_+ + (xs+y)\tr[B\id_{\{A>sB\}}] \\
        &= x \int 2(a-sb)_+ \, \di \mu + (xs+y) \int 2b \mathbf{1}_{\{a>sb\}} \, \di \mu \\
        &= 2\int (xa+yb) \mathbf{1}_{\{a>sb\}} \, \di \mu 
    \end{align}
    An analogous argument handles the reversed inequality.
\end{proof}

\begin{lemma}[(iii)]\label{lem:expdict}
Let $A,B\geq0$, $s>0$ and define $K_s[X] = \int_0^1 (1-r)\,\ee^{-rsX}\,\di r$ for Hermitian $X$.
Then
\begin{equation}\label{eq:expdict}
    \int_{\mathbb{R}^2} a\, \ee^{-sb} \meas = \tr\bigl[A\, K_s[B]\bigr] , \qquad
    \int_{\mathbb{R}^2} a\, \ee^{-sa} \meas = \tr\bigl[A\, K_s[A]\bigr] .
\end{equation}
\end{lemma}
\begin{proof}
Define $f_s(t) = \frac{1-\ee^{-st}}{s}$,
which fulfills the requirements of~\cref{thm:tracialjoint} and satisfies $f_s'(t)=\ee^{-st}$. Differentiation of its $H$-transform gives
\begin{equation}
    (H(f_s))'(z)=\int_0^1(1-r)\ee^{-rsz}\,\di r,
    \qquad (H(f_s))'(X)=K_s[X].
\end{equation}
So if we apply~\cref{thm:tracialjoint} with this function
to $B+t A$ and differentiate at $t=0$, we obtain the statement. On the scalar side, bounded support (see~\cref{prop:basic}) guarantees dominated convergence. On the trace side apply~\cref{lem:trace-deriv}.  Applying the same argument to
$(1+t)A$ proves the second identity.
\end{proof}

\section{Scalar identities}
\label{sec:appscalar}
In the following, we prove the scalar identities of~\cref{lem:scalar-integrals}.

\begin{proof}[Proof of~\cref{eq:rep1}]
    Define for $r>0$
    \begin{equation}
        \psi(r) = \int_0^1 \frac{(r-s)_-}{s} \, \di s + \int_1^\infty \frac{(r-s)_+}{s} \, \di s \, .
    \end{equation}
    Differentiating $\psi(r)$ gives $\psi'(r) = -\int_0^1 \frac{\mathbf{1}_{\{r<s\}}}{s} \di s + \int_1^\infty \frac{\mathbf{1}_{\{s<r\}}}{s} \di s = \log r$. Since $\psi(1)=0$, we get $\psi(r) = \int_1^r \log s \, \di s=r\log r -r +1$. Now take $r=a/b$ and multiply the expression by $b$. This leads to
    \begin{equation}
        a \log \frac{a}{b} = a-b+\int_0^1 \frac{(a-sb)_-}{s}\di s + \int_1^\infty \frac{(a-sb)_+}{s} \di s \,.
    \end{equation}
    If $a=0<b$, the first integral equals $b$ and the second vanishes. 
\end{proof}

\begin{proof}[Proof of~\cref{eq:rep2}]
    Substitute $t=1/s$ in the first integral of~\cref{eq:rep1}. This gives
    \begin{equation}
        \int_0^1 \frac{(a-sb)_-}{s} \, \di s  = \int_1^\infty \frac{(a-b/t)_-}{t} \, \di t = \int_1^\infty \frac{(b-ta)_+}{t^2} \, \di t \, .
    \end{equation}
\end{proof}
\begin{proof}[Proof of~\cref{eq:rep3}]
    Note, that for $a,b\geq 0$, the integrand vanishes for $s\in(0,1)$. To simplify upper and lower part, define $t=s/(s-1)$. Then for $s<0$, we get
    \begin{equation}
        \int_{-\infty}^0 \frac{\bigl((1-s)a+sb\bigr)_-}{|s|(s-1)^2} \, \di s = \int_0^1 \frac{(a-tb)_-}{t} \, \di t \, .
    \end{equation}
    by substitution. Similar, for $s>1$, we have
    \begin{equation}
        \int_{1}^\infty \frac{\bigl((1-s)a+sb\bigr)_-}{|s|(s-1)^2} \, \di s = \int_1^\infty \frac{(a-tb)_+}{t} \, \di t \, .
    \end{equation}
    \cref{eq:rep3} follows then from~\cref{eq:rep1}.
\end{proof}

\begin{proof}[Proof of~\cref{eq:rep4}]
    Let $a,b>0$ and define $\psi(t) = t\log t$. Then Taylor's formula gives us
    \begin{equation}
        \psi(a) - \psi(b) - \psi'(b)(a-b) = (a-b)^2 \int_0^1 (1-s) \psi''(b+s(a-b)) \, \di s \, .
    \end{equation} 
    Since $\psi'(b) = 1 + \log b$ and $\psi''(t)=1/t$ the claim follows. For $a=0<b$, the integral equals $b$ which makes the whole expression vanish.
\end{proof}

\begin{proof}[Proof of~\cref{eq:rep5}]
    For $a\geq0,b>0$, we have 
    \begin{equation}
        b\int_0^\infty (1+\log s) \mathbf{1}_{\{a>sb\}} \, \di s = b \int_0^{a/b}(1+\log s)\, \di s = a \log \frac{a}{b} \, .
    \end{equation}
\end{proof}

\begin{proof}[Proof of~\cref{eq:rep6}]
    Let $a>b>0$ and note that $\ee^{-sb} - \ee^{-sa} = \int_b^a s \ee^{-st} \, \di t$. By Fubini-Tonelli, we then get 
    \begin{equation}
        \int_0^\infty \frac{\ee^{-sb} - \ee^{-sa}}{s} \, \di s = \int_b^a \left(\int_0^\infty \, \ee^{-st} \di s \right)\, \di t = \int_b^a \frac{1}{t} \, \di t = \log \frac{a}{b}
    \end{equation}
    Fubini-Tonelli was applicable because the integrand was positive. For $0<a<b$, interchange $a$ and $b$ and change the sign. Multiplying by $a$ gives the result. The case $a=0<b$ leads to 0.
\end{proof}

\section{Testing geometry, curvature, and stability}
\label{sec:apptesting}
\subsection{Proof of the curvature criterion}
Let $X,Y\in M_n(\mathbb{C})$ be Hermitian and denote eigenvalues and vectors of $X$ as $\lambda_i$ and $u_i$, respectively.
Furthermore, let $g$ be twice continuously differentiable near the
spectrum of $X$. Then one has
\begin{equation}\label{eq:second-derivative}
    \left.\frac{\di^2}{\di s^2}\tr g(X+sY)\right|_{s=0}
    =\sum_{i,j}
      \frac{g'(\lambda_i)-g'(\lambda_j)}{\lambda_i-\lambda_j}
      |\langle u_i,Yu_j\rangle|^2,
\end{equation}
by differentiating the divided-difference formula also used in~\cref{lem:trace-deriv}. When $\lambda_i = \lambda_j$, the term is understood as $g''(\lambda_i).$

\begin{proof}[Proof of \cref{thm:curvature}]
By~\cref{eq:stoploss}, we have
\begin{equation}
    h(t) = \tr[\rho-t\sigma]_+ = \int_0^\infty (r-t)_+\,\di \nu(r) \, ,
\end{equation}
and taking right and left derivative gives by~\cref{thm:diffdict} $h'_+(t) = - \tr[\sigma \mathbf{1}_{\{\rho>t\sigma\}}]$ and $h'_-(t) = - \tr[\sigma \mathbf{1}_{\{\rho\geq t\sigma\}}]$. Therefore $h'_+(t)-h'_-(t) = \tr[\sigma P_{\ker(\rho-t\sigma)}]$ and since $\sigma>0$, this difference is strictly positive whenever $\rho-t\sigma$ is singular. Singular threshold-values are therefore corners of the testing curve. Because of this, in the following we focus on $t>0$ for which $X=\rho-t\sigma$ is invertible. 
Choose a smooth function $g$ which agrees with $x\mapsto x_+$ on a neighborhood of the spectrum of $X$. For sufficiently small $s$, the spectrum of $X-s\sigma$ remains in this neighborhood and we have $h(t+s) = \tr[g(X-s\sigma)]$. 
With $X= \sum_i \lambda_i \ket{u_i}\bra{u_i}$,~\cref{eq:second-derivative} gives
\begin{equation}\label{eq:curvature2}
    h''(t) = 2 \sum_{\lambda_i > 0 > \lambda_j} \frac{|\langle u_i,\sigma u_j \rangle|^2}{\lambda_i-\lambda_j} \, .
\end{equation}

Now, define the projector $P_t=\id_{\{\rho-t\sigma>0\}}$. Since every summand in~\cref{eq:curvature2} is nonnegative, $h_{\rho,\sigma}''(t)=0$ is equivalent to $P_t\sigma(\id-P_t)=0$.
Because $\sigma$ is Hermitian, this means that $P_t$ makes $\sigma$ block diagonal.
Trivially, it also block diagonalizes $\rho-t\sigma$, and hence it also block diagonalizes $\rho$. 

Thus, if $X$ has both positive and negative eigenvalues, zero curvature produces a nontrivial common block decomposition of $\rho$ and $\sigma$.

This allows us to decompose the full Hilbert space. Whenever a nontrivial orthogonal projection commutes
with both $\rho$ and $\sigma$, we split it. Finite dimension of the Hilbert spaces ensures that this process ends with an orthogonal direct sum decomposition of minimal common blocks. So $\rho = \bigoplus_k \rho_k$, $\sigma = \bigoplus_k \sigma_k$ and $h(t) = \sum_k h_k(t)$ with $h_k(t) = \tr[\rho_k-t \sigma_k]_+$.
Note, that each $h_k$ is convex since 
\begin{equation}
    h_k(t) = \max_{0\leq T \leq \id_k} \tr[T\rho_k] - t \, \tr[T\sigma_k]
\end{equation}
is a maximum of affine functions of t. 

Consider now a minimal block $\cH_k$ of dimension greater than one and define $C_k=\sigma_k^{-1/2}\rho_k\sigma_k^{-1/2}$.
Note, that $\rho_k-t \sigma_k$ is singular precisely when $t$ is an eigenvalue of $C_k$. 
Let $t_-$ and $t_+$ be smallest and largest eigenvalue of $C_k$. By assumption, they have to be different. Otherwise $C_k=c\id$ and therefore
$\rho_k=c\sigma_k$. Then the span of any eigenvector of $\sigma_k$ would be a smaller common block,
contradicting $\cH_k$ being minimal.

Now we show, that for any $t\in(t_-,t_+)$, the matrix $\rho_k - t\sigma_k$ has both positive and negative eigenvalues. This can be seen by choosing unit eigenvectors
$\ket{v_\pm}$ of
$C_k$ corresponding to $t_\pm$. Furthermore, define $\ket{w_\pm}=\sigma_k^{-1/2}\ket{v_\pm}$. Then
\begin{equation}
    \bra{w_\pm}\rho_k-t\sigma_k\ket{w_\pm} =\bra{v_\pm}C_k-t\id_k\ket{v_\pm} = t_\pm -t \, ,
\end{equation}
which clearly takes positive and negative values. This means that there exists a non-trivial positive and negative spectral projection.
If additionally $t$ is not an eigenvalue of $C_k$, then $h''_k(t)=0$ would imply that the positive spectral projection of $\rho_k-t\sigma_k$ gives a new nontrivial common block decomposition of $\rho_k$ and $\sigma_k$ which contradicts $H_k$ being minimal. Therefore $h_k''(t)>0$ for every $t\in(t_-,t_+)\setminus \spec(C_k)$ which makes it strictly convex.

We can now prove (i)$\Leftrightarrow$(ii). If $\rho$ and $\sigma$ commute, choose a common eigenbasis, with eigenvalues $p_i$ and $q_i>0$. Then $h_{\rho,\sigma}(t)=\sum_i(p_i-tq_i)_+$,
which has finitely many affine pieces.
Conversely, suppose that $h_{\rho,\sigma}$ has finitely many affine pieces. If a minimal common block had dimension greater than one, its contribution $h_k$ would be strictly convex on the nonempty interval $(t_-,t_+)$. All other block contributions are convex, so their sum $h$ would also be strictly convex on that interval. This is impossible for a function consisting of finitely many affine pieces. Hence all minimal common blocks are one-dimensional and thus give a common eigenbasis for $\rho$ and $\sigma$.

To prove (ii)$\Leftrightarrow$(iii), we recall that by the one-sided derivative of~\cref{eq:stoploss} seen earlier during the proof and in~\cref{rem:layercakef}, we have
\begin{equation}
    \nu[(a,b)] = \nu[(a,\infty)] - \nu([b,\infty)] = h'_-(b) - h'_+(a)
\end{equation}
which means that if $h$ is affine with slope $m$ on an open interval $I$, then, for every $a<b$ in $I$, $\nu[(a,b)]=m-m=0$ and therefore $\nu(I)=0$. On the other hand, if $\nu(I)=0$, then $\nu[(a,b]) = 0 = h'_+(b) - h'_+(a)$ for all $a<b$ in $ I$, which means that $h'_+$ is constant and $h$ affine there. Now suppose that $h$ has finitely many affine pieces with breakpoints $0<t_1<\dots<t_N<\infty$. Then the previous argument shows that $\nu[ [0,\infty)\setminus \{0,t_1,\dots,t_N\}] = 0$ and therefore that it has finite support. For the other direction, suppose that $\nu$ has finite support and that we therefore can write $\nu = \sum_i v_i \delta_{r_i}$. Then by~\cref{eq:stoploss}, we have $h(t) = \sum_i v_i (r_i-t)_+$, which has finitely many affine pieces.

We continue with (i)$\implies$(iv). If $\rho$ and $\sigma$ commute, then a quantum test $T$ contributes to the testing probabilites only through its diagonal entries
$t_i=\bra{i}T\ket{i}\in[0,1]$. Assuming eigenvalues $p_i$ and $q_i$ for $\rho$ and $\sigma$, every choice of $t_i$ defines a valid diagonal test and we get $\cT_{\rm q}(\rho,\sigma) = \{ \sum_i t_i (p_i,q_i):0\leq t_i \leq 1 \}$, which is a polygon.

We end the proof with (iv)$\Rightarrow$(ii). By definition, every point in $\cT_{\rm q}(\rho,\sigma)$ is of the form $(p,q) = (\tr[T\rho],\tr[T\sigma])$ and therefore 
\begin{equation}
    h(t) = \max_{(p,q)\in\cT_{\rm q}(\rho,\sigma)} (p-tq) .
\end{equation}
If the testing region is a polygon with $N$ vertices $v_i=(a_i,b_i)$, then maximizing $p-tq$ over it gives 
\begin{equation}
    h(t) = \max_{1\leq i \leq N} (a_i-tb_i)
\end{equation}
which is a maximum over finitely many affine functions, which means $h$ has finitely many affine pieces.
\end{proof}

\bibliographystyle{arxiv_no_month}
\bibliography{bibliofile}

\end{document}